\documentclass[conference,letterpaper]{IEEEtran}

\usepackage[utf8]{inputenc} 
\usepackage[T1]{fontenc}
\usepackage{url}
\usepackage{ifthen}
\usepackage{cite}
\usepackage[cmex10]{amsmath}
\usepackage{amsthm, latexsym, amssymb, amsfonts, epsf, xcolor, hyperref, mathrsfs,graphicx}
\usepackage{algorithmic}
\usepackage{amsmath, amsthm, latexsym, amssymb, amsfonts, epsf, xcolor, hyperref}
\usepackage{caption, subcaption, enumerate, color}
\usepackage{multirow,multicol}
\usepackage{mathtools}
\usepackage{bm}
\usepackage{blkarray}
\usepackage{booktabs}
\usepackage{framed}
\usepackage{setspace}
\usepackage{tikz}
\usepackage{graphicx}
\usepackage{wrapfig}
\usepackage{float}
\usepackage[dvips]{epsfig}
\usepackage{bbm}
\usepackage{url}
\usepackage{booktabs}
\usepackage{bibentry}
\usepackage[normalem]{ulem}
\usepackage{comment}
\usepackage{lingmacros}
\usepackage{cleveref}
\usepackage{extarrows}
\usepackage{multirow}
\usepackage{mathtools}
\usepackage{enumerate}
\usepackage{array}
\usepackage{bigstrut}
\usepackage{exscale,relsize}
\usepackage{graphicx}
\usepackage[utf8]{inputenc}
\usepackage[linesnumbered, ruled]{algorithm2e}

\newcommand{\1}{\mathbbm{1}}

\newcommand{\Rows}{\operatorname{Rows}}

\newcommand{\F}{{\mathbb F}}

\numberwithin{equation}{section}
\newtheorem{theorem}{Theorem}[section]
\newtheorem{lemma}[theorem]{Lemma}
\newtheorem{proposition}[theorem]{Proposition}
\newtheorem{corollary}[theorem]{Corollary}

\theoremstyle{definition}
\newtheorem{definition}[theorem]{Definition} 
\newtheorem{remark}[theorem]{Remark}
\newtheorem{example}[theorem]{Example}

\newcommand{\rmv}[1]{}

\DeclareMathOperator{\hull}{Hull}

\DeclareMathOperator{\Span}{Span}
\DeclareMathOperator{\Rowsp}{Rowsp}

\newcommand{\re}[1]{(\ref{e:#1})}

\theoremstyle{empty}

\title{Binary Triorthogonal and CSS-T Codes for Quantum Error Correction}

\makeatletter
\newcommand{\linebreakand}{%
  \end{@IEEEauthorhalign}
  \hfill\mbox{}\par
  \mbox{}\hfill\begin{@IEEEauthorhalign}
}
\makeatother

\author{
  \IEEEauthorblockN{Eduardo Camps-Moreno}
  \IEEEauthorblockA{\textit{Department of Mathematics} \\
    \textit{Virginia Tech}\\
    Blacksburg, VA USA \\
    e.camps@vt.edu}
  \and
  \IEEEauthorblockN{Hiram H. López}
  \IEEEauthorblockA{\textit{Department of Mathematics} \\
    \textit{Virginia Tech}\\
    Blacksburg, VA USA \\
   hhlopez@vt.edu}
    \and 
  \IEEEauthorblockN{Gretchen L. Matthews}
  \IEEEauthorblockA{\textit{Department of Mathematics} \\
    \textit{Virginia Tech}\\
    Blacksburg, VA USA \\
    gmatthews@vt.edu}
  \linebreakand 
\IEEEauthorblockN{Diego Ruano}
  \IEEEauthorblockA{\textit{
  IMUVA-Mathematics Research Institute} \\
    \textit{Universidad de Valladolid}\\
    Valladolid, Spain \\
    diego.ruano@uva.es}
\and 
\IEEEauthorblockN{Rodrigo San-José}
  \IEEEauthorblockA{\textit{
  IMUVA-Mathematics Research Institute} \\
    \textit{Universidad de Valladolid}\\
    Valladolid, Spain \\
    rodrigo.san-jose@uva.es}
  \linebreakand
  \IEEEauthorblockN{Ivan Soprunov}
  \IEEEauthorblockA{\textit{Department of Mathematics and Statistics} \\
    \textit{Cleveland State University}\\
    Cleveland, OH, USA \\
    i.soprunov@csuohio.edu}
}

\begin{document}

\maketitle

\begin{abstract}
In this paper, we study binary triorthogonal codes and their relation to CSS-T quantum codes. We characterize the binary triorthogonal codes that are minimal or maximal with respect to the CSS-T poset, and we also study how to derive new triorthogonal matrices from existing ones. Given a binary triorthogonal matrix, we characterize which of its equivalent matrices are also triorthogonal. As a consequence, we show that a binary triorthogonal matrix uniquely determines the parameters of the corresponding triorthogonal quantum code, meaning that any other equivalent matrix that is also triorthogonal gives rise to a triorthogonal quantum code with the same parameters.
\end{abstract}

\begin{IEEEkeywords}
Triorthogonal codes, quantum codes, CSS construction, linear codes
\end{IEEEkeywords}

\section{Introduction} \label{sec:intro}

Due to noise and decoherence, quantum error-correction is required to achieve quantum fault-tolerant computation. The most well-known construction of quantum error-correcting codes (QECCs) is the CSS construction, obtained independently in the works of Calderbank and Shor \cite{CalderbankShor_96} and Steane \cite{Steane_96}. Since this construction requires a nested pair of classical codes $C_2\subseteq C_1$, it provides a bridge between QECCs and classical coding theory.

Due to Eastin–Knill theorem \cite{eastinknill}, it is not possible to find a QECC that implements a universal gate set transversely. One strategy to circumvent this limitation is to use magic state distillation protocols to implement a logical non-Clifford gate, usually the $T$ gate \cite{bravyikitaevMagicStateDistillation}, which requires QECCs implementing the logical $T$ gate. In this context, triorthogonal codes and CSS-T codes have gained attention because of their potential for magic state distillation. Triorthogonal matrices were introduced by Bravyi and Haah \cite{bravyiTriorthogonalOriginal} as binary matrices in which the common supports of all pairs and of all triples of rows have even cardinalities. From such a matrix, one can construct a QECC that implements the logical $T$ gate when applying a physical transversal $T$ gate, up to a possible Clifford correction. In fact, if one wants to avoid the Clifford correction, then \cite{rengaswamyOptimalityCSST} shows that triorthogonality, plus some weight conditions, is a necessary and sufficient condition to obtain the logical $T$ from the physical transversal $T$ gate. 

CSS-T codes are a generalization of triorthogonal codes introduced in \cite{rengaswamyOptimalityCSST}. These codes only require the physical transversal $T$ gate to induce some logical operation on the logical qubits, which might not be the logical $T$. In \cite{CLMRSS_CSST_24}, an alternative characterization of binary CSS-T codes is given, and the poset of binary CSS-T pairs is introduced. By definition, triorthogonal codes are CSS-T codes, and some of the connections between both types of codes are studied in \cite{rengaswamyOptimalityCSST}. In this paper, we further explore the relations between binary triorthogonal and binary CSS-T quantum codes, and we study how binary triorthogonal codes fit within the binary CSS-T poset.

The main contributions of this paper include 
\begin{itemize}
    \item a propagation rule for triorthogonal codes: Theorem \ref{T:step} and Corollary \ref{C:prop} shows that if $C$ is 
    linear code giving rise to a triorthogonal $[[n,k,d]]$ QECC and $v\in(C^{\ast 2})^\perp \setminus C$ is a vector of odd weight, then $C+\langle v\rangle$ yields a triorthogonal $[[n,k+1,d]]$ QECC via the CSS construction.
    \item a description of the poset of 
linear codes giving rise to triorthogonal QECCs via the CSS construction (see Definition \ref{def:tri_poset}) with minimal and maximal elements given in Theorem \ref{T:min_max}.
\end{itemize}

This paper is organized as follows. Section \ref{S:prelim} covers the necessary background. Section \ref{S:poset} discusses the poset of triorthogonal codes and some relations to the CSS-T poset. A conclusion is provided in Section \ref{S:conclusion}.

\section{Preliminaries}
\label{S:prelim}

We use the standard notation for finite fields and matrices:
$\F_2$ denotes the finite field with two elements $0$ and $1$; $\F_2^{m \times n}$ denotes the set of $m \times n$ matrices with entries in $\F_2$; and $\F_2^n:=\F_2^{1 \times n}$. Given $u, v \in \F_2^n$, $u \cdot v:=\sum_{i=1}^n u_iv_i \in \F_2$ denotes their usual dot product. The weight of a vector $v \in \F_2^n$ is taken to be its Hamming weight $wt(v):=\sum_{i=1}^n v_i$. Sometimes we write $|v|=wt(v)\mod 2$, or, equivalently, $|v|=\langle v,v\rangle$.  The set of rows of a matrix $M \in \F_2^{m \times n}$ is denoted by $\Rows(M)$, the rowspace of $M$ is denoted by $\Rowsp(M)$, and the transpose of $M$ is $M^t \in \F_2^{n \times m}$. The all-ones vector is $\1:=(1,\ldots,1) \in \F_2^n$, and the $k \times k$ identity matrix is denoted by $I_k$. We also use the standard notation from coding theory. A binary linear code $C$ of length $n$, dimension $k$, and minimum distance $d$ is a $k$-dimensional $\F_2$-subspace of $\F_2^n$ in which the minimum Hamming weight of a nonzero codeword is $d$, and referred to as an $[n,k,d]$ code.  The dual of such a code $C$ is an $[n,n-k,d']$ code $C^{\perp}:=\left\{ u \in \F_2: u \cdot c = 0 \ \forall c \in C \right\}$. A generator matrix for $C$ is a matrix $G \in \F_2^{k \times n}$ where $n \geq k$ with $\Rowsp(G)=C$. 
In this paper, we restrict our attention to binary linear codes and, hence, say code to mean a binary linear code. The relative hull of a code $C_1$ with respect to another code $C_2$ of the same length is $$\hull_{C_2}(C_1):=C_1 \cap C_2^{\perp},$$ 
and the hull of a code $C$ is the relative hull of $C$ with respect to itself:
$$\hull(C):=C \cap C^{\perp}.$$ The Schur product of vectors $u=(u_1,\ldots,u_n)$ and $v=(v_1,\ldots,v_n)$ in $\F_2^n$ is the vector
$$u \star v := (u_1v_1,\ldots, u_nv_n) \in \F_2^n.$$ 
The Schur product of codes $C, C' \subseteq \F_2^n$ is the code 
$$
C \star C':=\Span \left\{ c \star c': c \in C, c' \in C'\right\} \subseteq \F_2^n.
$$
The square of the code $C$ is $C^{\star 2}:=C \star C$, and for a positive integer $i$, the $i$-th power of $C$ is $C^{\star i}:=\underbrace{C \star \cdots \star C}_{i \ \textnormal{times}}$. 

Codes $C_1$ and $C_2$ satisfying $C_2 \subseteq C_1$ may be used to define a  quantum stabilizer code $Q(C_1,C_2)$ via the CSS construction \cite{CalderbankShor_96, Steane_96}. 
The CSS code $Q(C_1,C_2)$ is an
$$[[n,k_1-k_2,\ge \min\{d_1,d_2^\perp\}]]$$ quantum code, where $C_i$ is an $[n,k_i,d_k]$ code, for $i=1,2$, and $d_2^\perp$ is the minimum distance of $C_2^\perp$. It was shown in \cite[Theorem 2.3]{CLMRSS_CSST_24} that such a code is CSS-T if and only if 
\begin{equation} \label{eq:CSST}
C_2 \subseteq \hull_{C_1^{\star 2}}C_1=C_1 \cap \left(C_1^{\star 2}\right)^{\perp}. 
\end{equation}
In this case, we refer to $(C_1,C_2)$ as a CSS-T pair. Moreover, we can say more about the parameters of the resulting quantum code. 

\begin{proposition} \cite[Corollary 2.5]{CLMRSS_CSST_24} \label{prop:CSST_params}
If $(C_1,C_2)$ is a CSS-T pair, then $Q(C_1,C_2)$ is an $[[n,k_1-k_2,\ge d_2^\perp]]$ code.  
\end{proposition}

\section{Poset of triorthogonal codes}
\label{S:poset}

In this section, we consider binary triorthogonal matrices and the codes they define. 

\begin{definition}
A binary matrix $G$ is called {\it triorthogonal} if for every triple of distinct $u,v,w\in\Rows(G)$ the Schur products $u\star v$ and $u\star v\star w$ have even weights.
A binary linear code $C\subset\F_2^n$ is called {\it triorthogonal} if it has a triorthogonal generator matrix.
\end{definition}

\begin{remark} \label{rem:matrix_rep}
Note that a binary linear code may have several triorthogonal generator matrices $G$. However, the span of the even weighted rows of $G$ is unique and equals the hull of $C$, as shown in \cite[Lemma 1]{bravyiTriorthogonalOriginal};
that is, given any triorthogonal generator matrix $G \in \F_2^{k \times n}$ of a code $C$, one has
$$
\Rowsp \left( G_0 \right) =\hull(C),
$$
regardless of the choice of $G$, where $G_0 \in \F_2^{k' \times n}$ denotes the submatrix of $G$ consisting of all even weighted rows of $G$. We do
not distinguish between matrices with the same set of rows, only permuted, as they generate the same code. Hence, we may assume that given a generator matrix $G$, its rows are ordered so that 
\begin{equation}\label{e:trimat}
G=\begin{pmatrix}
G_1\\
G_0
\end{pmatrix}
\end{equation}
where the rows of $G_1$ are all of odd weight and the rows of $G_0$ are all of even weight.
\end{remark}

As mentioned above, CSS-T codes are a generalization of triorthogonal codes. This can be concluded from the following proposition.

\begin{proposition}\label{prop:tri_is_csst}
Let $C$ be a binary triorthogonal code. Then $$C^{\star 2}\subseteq C+C^\perp.$$
As a consequence, if we let $C_1=C$, $C_2=\hull(C)$, then $(C_1,C_2)$ is a CSS-T pair.
\end{proposition}
\begin{proof}
Let $G$ be a triorthogonal generator matrix for $C$. Let $x=\sum_{i=1}^k u_i\star v_i$ be an element of $C^{\star 2}$ for some $u_i,v_i\in \Rows(G)$. If $u_i=v_i$ then $u_i\star v_i=u_i\star u_i=u_i\in C$. If $u_i\neq v_i$ then for any $w\in C$ the vector $u_i\star v_i\star w$ is even weighted and, hence, $(u_i\star v_i)\cdot w=0$. This shows that $u_i\star v_i\in C^\perp$. Therefore, $x\in C+C^\perp$.

Let $C_1$ and $C_2$ be as in the statement of the proposition. Then $(C_1^{\star 2})^\perp\supseteq \hull(C_1)$ by the previous reasoning, which implies
$$
\hull_{C_1^{\star 2}}(C_1)\supseteq \hull(C_1)=C_2.
$$
We finish by recalling Equation (\ref{eq:CSST}).
\end{proof}

Before continuing, let us clarify that Proposition \ref{prop:tri_is_csst} does not characterize triorthogonal codes. The next example demonstrates that the condition $C^{\star 2}\subseteq C+C^\perp$ (in fact, even  $C^{\star 2} = C+C^\perp$) is not sufficient to guarantee that $C$ is triorthogonal.

\begin{example}\rm
   Consider the binary code $C$ with generator matrix $G=\begin{pmatrix} 1&1&0&0&0\\ 0&1&1&0&0\\ 0&0&0&1&1\end{pmatrix}$. Then $C^\perp$ and $C^{\ast 2}$ are generated respectively by $$\begin{pmatrix} 1&1&1&0&0\\ 0&0&0&1&1\end{pmatrix}\ \ \text{and}\ \ \begin{pmatrix}1&0&0&0&0\\ 0&1&0&0&0\\ 0&0&1&0&0\\ 0&0&0&1&1\end{pmatrix}.$$

    From this, it is easy to see that $C^{\ast 2}=C+C^\perp$. However, $C$ cannot be generated by a triorthogonal matrix, since $C$ is even but $$\hull(C)=\mathrm{Span}\{\begin{pmatrix} 0&0&0&1&1\end{pmatrix}\}.$$ This  contradicts the condition given in Remark~\ref{rem:matrix_rep}. 
\end{example}

If $G$ is a generator matrix of a binary triorthogonal code $C$, then the associated CSS-T pair is $(C,\hull(C))$. Thus, we can speak about the triorthogonal code $C$ generated by $G$ or we can speak about the triorthogonal pair $(C,\hull(C))$.

\begin{definition}\rm
    A {\it quantum triorthogonal code} $Q$ is the quantum code obtained using the CSS-T pair $(C,\hull(C))$, where $C$ is a triorthogonal code.
\end{definition}

\begin{remark}\label{rem:hullc2}
Let $C$ be a binary triorthogonal code with generator matrix as in Equation (\ref{e:trimat}). By Remark \ref{rem:matrix_rep} we have $\Rowsp(G_0)=\hull(C)\supseteq \hull_{C^{\star 2}}(C)$. Since $(C,\hull(C))$ is a CSS-T pair, from Equation (\ref{eq:CSST}) we obtain
$$
\Rowsp(G_0)=\hull(C)=\hull_{C^{\star 2}}(C).
$$
\end{remark}

It is important to observe that triorthogonality is a property of the code rather than a property of a specific generator matrix. This suggests then that the triorthogonal matrix that generates a triorthogonal code is unique up to specific transformations. This is indeed the case, according to the following Theorem \ref{T:matrix_form} below.

First, we require a couple of lemmas.

\begin{lemma}\label{lemma:duh}
    Let $G$ be a binary triorthogonal matrix as in \re{trimat}. Then $$GG^t=\begin{pmatrix} I_{k_1}&0\\ 0&0\end{pmatrix},$$ or equivalently $G_1G_1^t=I_{k_1}$ and $G_iG_0^t=0$ for $i=0,1$.
\end{lemma}

\begin{proof}
    It follows from the definition of triorthogonality.
\end{proof}

\begin{lemma}\label{lemma:nosingulartri}
Let $G\in \F_2^{n\times n}$. If $G$ is nonsingular and triorthogonal, then $G$ is a permutation matrix. 
\end{lemma}
\begin{proof}
Let $g_i$, $1\leq i \leq n$, be the rows of $G$. Since $G$ is nonsingular, its rows generate $\F_2^n$ and thus, the hull is zero (Remark \ref{rem:matrix_rep}), implying that $g_i$ is odd for each $i$. Even more, for $i\neq j$, we have
$$
g_i\star g_j=\sum_{h=1}^n\alpha_h g_h,
$$
for some $\alpha_1,\dots,\alpha_n\in \F_2$. Since $G$ is triorthogonal, we obtain
$$
0=|g_i\star g_j\star g_k|=\left|\sum_{h=1}^n\alpha_h g_k\star g_h\right |=|\alpha_k g_k|
$$
where $|v|$ denotes the weight modulo 2. Since each $g_k$ is odd, we have $\alpha_k=0$. Hence, $g_i\star g_j=0$, which implies the rows of $G$ have pairwise disjoint support, proving the statement.
\end{proof}

\begin{theorem} \label{T:matrix_form}
Let $G_i\in \F_2^{k_i\times n}$, $i=0,1$, where $G_1$ has rows of odd weight and $G_0$ has rows of even weight. 
Suppose $G$ as in Expression \re{trimat} forms a triorthogonal matrix,  and let $M\in\F_2^{(k_0+k_1)\times (k_0+k_1)}$ be a non-singular matrix. Then $MG$ is a triorthogonal matrix if and only if
$$
M=\begin{pmatrix}
    P&M_2\\
    0&M_4
\end{pmatrix},
$$
where $P\in\F_2^{k_1\times k_1}$ is a permutation matrix,  $M_4\in\F_2^{k_0\times k_0}$ is non-singular, and $M_2\in \F_2^{k_1\times k_0}$. 
\end{theorem}

\begin{proof}
The ``if'' statement can be easily checked. Now, assume $$M=\begin{pmatrix} M_1&M_2\\ M_3&M_4\end{pmatrix}$$ is non-singular, where $M_1\in\F^{k_1\times k_1}$, $M_4\in\F^{k_0\times k_0}$. We have
$$
MG=\begin{pmatrix}
    M_1G_1+M_2G_0\\
    M_3G_1+M_4G_0
\end{pmatrix}.
$$
By Lemma \ref{lemma:duh}, we have $G_1G_1^t=I_{k_1}$ and $G_iG_0^t=0$, for $i=0,1$. Thus
$$MG(MG)^t=\begin{pmatrix}
    M_1M_1^t & M_1M_3^t\\
    M_3M_1^t & M_3M_3^t
\end{pmatrix}.$$

If $MG$ is also triorthogonal, by Lemma \ref{lemma:duh}, we have
$$
MG(MG)^t=\begin{pmatrix}
    M_1M_1^t & M_1M_3^t\\
    M_3M_1^t & M_3M_3^t
\end{pmatrix}=\begin{pmatrix}
    I_{k_1}&0\\
    0&0
\end{pmatrix},
$$
permuting rows if necessary (which is allowed according to Remark \ref{rem:matrix_rep}). 

Since $M_1M_1^t=I_{k_1}$, $M_1$ is invertible, and from $M_3M_1^t=0$ we deduce $M_3=0$. Since $M$ is invertible, we also have $\det(M_4)\neq 0$. The only thing left to prove is that $M_1$ is a permutation matrix. Let $g_i$ be the $i$-th row of $G$ and observe that the $i$-th row of $MG$, $1\leq i\leq k_1$ is of the form
$$g'_i=\sum_{j=1}^{k_1} (M_1)_{ij}g_j+\sum_{j=1}^{k_0} (M_2)_{ij}g_{k_1+j}.$$

Observe that for $1\leq i_1,i_2,i_3\leq k_1$ we have
\begin{align*}
&|g'_{i_1}\star g'_{i_2}\star g'_{i_3}|\\
=&\left|\sum_{j_1=1}^{k_1}\sum_{j_2=1}^{k_1}\sum_{j_3=1}^{k_1}\left( \prod_{h=1}^3 (M_1)_{i_h j_h} \right) g_{j_1}\star g_{j_2}\star g_{j_3}\right|,
\end{align*}

\noindent where the products with at least one row of $G_0$ are not relevant since $\mathrm{Rowsp}(G_0)=\mathrm{Hull}_{C^{\ast 2}}(C)$ by Remark \ref{rem:hullc2}, where $C$ is the span of $G$. Since $G$ is triorthogonal, then $g_{j_1}\star g_{j_2}\star g_{j_3}$ is even if at least two of the $j_i$'s are different, thus

\begin{align*}
|g'_{i_1}\star g'_{i_2}\star g'_{i_3}|&=\left|\sum_{j=1}^{k_1} \prod_{h=1}^3 (M_1)_{i_hj}g_j\right| \\
&=| (M_1)_{i_1}\star (M_1)_{i_2}\star (M_1)_{i_3}|,
\end{align*}
where we have used that $|g_j|=1$ and the fact that the weight modulo 2 is a linear map from $\F_2^n$ to $\F_2$.

Therefore, $MG$ is triorthogonal if and only if $M_1$ is triorthogonal and $\det(M_4)\neq 0$. Since $M_1$ is non-singular then it is a permutation matrix by Lemma \ref{lemma:nosingulartri} and we have the conclusion.
\end{proof}

\begin{remark}
    Notice that Theorem \ref{T:matrix_form} also highlights the fact that a binary triorthogonal code may have a generator matrix which is not a binary triorthogonal matrix. 
\end{remark}


Theorem \ref{T:matrix_form} proves that a triorthogonal basis of $C$ (in the sense that the elements of the basis form a triorthogonal matrix) is unique modulo $\hull(C)$. As a corollary we have the following.

\begin{corollary}
    Let $C\subseteq\mathbb{F}_2^n$ be a binary triorthogonal code. Then the corresponding quantum triorthogonal code has parameters $$[[n,\dim C-\dim\hull(C),\geq d(C+C^\perp)]]_2.$$
\end{corollary}

\begin{remark}\rm
    The $X$-stabilizers of a quantum triorthogonal code are given by $\hull(C)$ while the $Z$-stabilizers are given by $C^\perp$. Once again, the choice of a triorthogonal matrix for $C$ is not relevant, since any of such matrices gives the same quantum code.
\end{remark}

In analogy with CSS-T pairs, we will see that we can either increase or reduce the dimension of a triorthogonal code.

\begin{proposition}
    Let $C$ be a binary triorthogonal code of dimension at least $2$. There is a triorthogonal code $C'$ such that $\dim C=\dim C'+1$.
\end{proposition}

\begin{proof}
   This follows immediately by deleting an odd row in a triorthogonal generator matrix of $C$.
\end{proof}

\begin{theorem}\label{T:step}
Let $C$ be a binary triorthogonal code. If $v\in(C^{\ast 2})^\perp$, then $C+\langle v\rangle$ is a triorthogonal code.
\end{theorem}

\begin{proof}
By hypothesis, $C$ has a triorthogonal generator matrix 
$$
G=\begin{pmatrix}
G_1\\
G_0
\end{pmatrix}.
$$
We consider the matrix
$$
G'=\begin{pmatrix}
G'_1\\
G'_0    
\end{pmatrix},
$$
where $G'_1=\begin{pmatrix}
    v\\
    G_1
\end{pmatrix}$ and $G'_0=G_0$ if $v$ is odd weighted; and $G'_1=G_1$ and $G_0=\begin{pmatrix}
    v\\
    G_0
\end{pmatrix}$ if $v$ is even weighted.

We have that $G'_1$ has odd weighted rows and $G'_0$ has even weighted rows. We will prove that $G'$ is triorthogonal. The only conditions we need to check are those that involve $v$ since $G$ is triorthogonal. Let $g\neq v$ be a row of $G'$. Then $g\in C$ and $v\cdot g=0$ since $v\in (C^{\star 2})^\perp\subset C^\perp $ (recall $C\subset C^{\star 2}$ as $C$ is a binary code). Let $g_1\neq g_2$ be two rows of $G'$ different from $v$. Then $g_1\star g_2\in C^{\star 2}$, which implies $(g_1\star g_2)\cdot v=0$, proving that $G'$ is triorthogonal. 
\end{proof}

\begin{corollary} \label{C:prop}
    Let $C$ be a triorthogonal binary linear code, $v\in (C^{\ast 2})^\perp$ and denote by $C'=C+\langle v\rangle$. Let $Q$ and $Q'$ be the corresponding quantum triorthogonal codes. If $v$ is odd and $Q$ is $[[n,k,d]]$, then $Q'$ is a binary $[[n,k+1,d]]$ triorthogonal code. 
\end{corollary}
\begin{proof}
The result follows from Theorem \ref{T:step} and the fact that we are changing $C$ without changing its hull. Thus, the bound on the minimum distance remains the same.     
\end{proof}
\begin{remark}\rm
    Keeping the same notation as in Corollary \ref{C:prop} above, if $v$ is even, $\dim Q=\dim Q'$ but the lower bound on the minimum distance of the codes can be different since $d(C+C^\perp)\leq d(C'+{C'}^\perp)$. 
\end{remark}

Recall that the poset of binary CSS-T pairs $(C_1,C_2)$
is defined using the entry-wise partial order on the pairs, i.e. $(C_1,C_2)\leq (C_1',C_2')$ if and only if
$C_1\subseteq C_1'$ and $C_2\subseteq C_2'$; see \cite{CLMRSS_CSST_24} for additional details. As any trigorthogonal code $C$
defines a CSS-T pair $(C,\hull(C))$, we
propose the following definition of a poset of binary triorthogonal codes.

\begin{definition} \label{def:tri_poset}
Let $C$ and $C'$ be triorthogonal codes in $\F_2^n$. Define a partial order 
$$C\leq C'$$ if and only if $$C\subseteq C'\ \textnormal{ and }\ \hull(C) \subseteq \hull(C').$$
Denote by $\cal{T}$ the poset of triorthogonal codes under this partial order.
\end{definition}

The next result characterizes the minimal and maximal elements of $\cal{T}$.
\rmv{Recall that a code $C$ is said to be {\it isotropic} if  $C^{\star 3}\subseteq \left< \1 \right>^\perp$. }

\begin{theorem} \label{T:min_max}
Let $C$ be a binary triorthogonal code.
\begin{enumerate}
    \item $C$ is a minimal element of $\cal{T}$ if and only $C$ is a one-dimensional even code.
    \item $C$ is a maximal element of $\cal{T}$ if and only if 
$(C^{\star 2})^\perp=\hull(C)$.
\end{enumerate}
\end{theorem}

\begin{proof} 

1) This result follows immediately from the definition. 

2) 
The code $C$ is maximal if and only if $(C^{\star 2})^\perp\subseteq C$ (one direction is clear by Theorem \ref{T:step}, and the other one follows from the definition of triorthogonal codes). Since $C$ is binary, $C\subseteq C^{\star 2}$, i.e., $(C^{\star 2})^\perp \subseteq C^\perp$. We have $(C^{\star 2})^\perp \subseteq \hull(C)=\hull_{C^{\star 2}}(C)$ by Remark \ref{rem:hullc2}, which implies $(C^{\star 2})^\perp=\hull_{C^{\star 2}}(C)=\hull(C)$.
\end{proof}

\rmv{\begin{theorem}
\rsj{I think there was something similar to the previous result for the $p$-ary case, with $v\in (C_1^{\star 2})^\perp\cap C_1^\perp$. Probably the same proof works}
\end{theorem}}

As a Corollary of Theorem \ref{T:step}, we can obtain the following result about the poset of triorthogonal codes.

\begin{corollary} 
Let $C$ be a binary triorthogonal code. If $\1\in C^{\star 2}$, then $C$ cannot be extended to another triorthogonal code without modifying the hull.
\end{corollary}
\begin{proof}
To extend $C$ to another triorthogonal code we need to find vectors in $(C^{\star 2})^\perp$. If $\1\in C^{\star 2}$, all of these vectors are even weighted. The argument from Theorem \ref{T:step} for the even weighted case shows that resulting code after extending has increased the dimension of its hull.
\end{proof}

In \cite{CLMRSS_CSST_24} it is proven that a binary CSS-T pair $(C_1,C_2)$ is maximal in the second component if and only if $C_2=C_1\cap (C_1^{\star 2})^\perp$. This is always the case with triorthogonal quantum codes $(C,\hull(C))$ by Remark \ref{rem:hullc2}. Therefore, we obtain the following.

\begin{corollary}
If $C$ is a binary triorthogonal code, then the CSS-T pair $(C,\hull(C))$ is maximal in the second component with respect to the CSS-T poset.
\end{corollary}

\section{Conclusion} \label{S:conclusion}
In this paper we have studied binary triorthogonal codes and their poset. We have also shown that equivalent triorthogonal matrices give rise to quantum codes with the same parameters and we have given a propagation rule for quantum triorthogonal codes. Future research agenda includes constructing triorthogonal codes with good parameters using well known families of classical codes, such as cyclic or quasi-cyclic codes, and studying the generalization of triorthogonal codes to the $p$-ary case \cite{tillichTriorthogonalpary}.

\section{Acknowledgements}
Part of this work was done during the visit of Eduardo Camps Moreno, Hiram H. L\'opez, Gretchen L. Matthews, and Ivan Soprunov to Universidad de Valladolid. They thank Diego Ruano and Rodrigo San-Jos\'e for their hospitality. The initial collaboration amongst the group (absent San-Jos\'e) was facilitated by the Collaborate@ICERM program, supported by the National Science Foundation under Grant No. DMS-1929284.

Hiram H. L\'opez was partially supported by the NSF grant DMS-2401558.
Gretchen L. Matthews was partially supported by NSF DMS-2201075 and the Commonwealth Cyber Initiative.
Diego Ruano and Rodrigo San-Jos\'e were partially supported by Grant TED2021-130358B-I00 funded by MICIU/AEI/ 10.13039/501100011033 and by the ``European Union NextGenerationEU/PRTR'', by Grant PID2022-138906NB-C21 funded by MICIU/AEI/ 10.13039/501100011033 and by ERDF/EU, and by Grant QCAYLE supported by the European Union.-Next Generation UE/MICIU/PRTR/JCyL. Rodrigo San-Jos\'e was also partially supported by Grant FPU20/01311 funded by the Spanish Ministry of Universities.

\section*{Declarations}
\subsection*{Conflict of interest} The authors declare no conflict of interest.

\bibliographystyle{abbrv}
\bibliography{Allerton}

\end{document}